%% file: mixed-layouts-arxiv.tex
\documentclass[a4paper,USenglish]{llncs}
\usepackage[utf8]{inputenc}

\newif\ifproc
\proctrue

\usepackage{graphicx}
\usepackage{hyperref}
\usepackage{amssymb, amsmath}
\graphicspath{{figures/}}

\renewcommand{\paragraph}[1]{\smallskip\noindent\textbf{\textsf{#1}}}

\usepackage{csquotes}
\usepackage{subcaption}
\usepackage{todonotes}

\usepackage{xcolor}
\usepackage{xspace}
\usepackage{booktabs}
\usepackage{colortbl}
\colorlet{tablerowcolor}{gray!10}

\usepackage[inline]{enumitem}
\setlist[enumerate]{nosep}
\usepackage{hyperref}
\usepackage{timetravel}
\setCounters{theorem,lemma,corollary}
\usepackage{wrapfig}
\usepackage{multirow}
\usepackage[misc,geometry]{ifsym} 

\newcommand{\df}[1]{\emph{#1}}

\spnewtheorem*{proofWithFormula}{Proof}{\itshape}{\rmfamily}
\spnewtheorem*{sketchWithFormula}{Proof sketch}{\itshape}{\rmfamily}
\spnewtheorem*{sketch}{Proof sketch}{\itshape}{\rmfamily}

\spnewtheorem{myclaim}{Claim}{\bfseries}{\itshape}
\spnewtheorem*{myproof}{Proof}{\itshape}{\rmfamily}

\let\doendproof\endproof
\renewcommand\endproof{~\hfill$\qed$\doendproof}

\newcounter{casecounter}
\newcounter{subcasecounter}
\newcounter{subsubcasecounter}
\makeatletter
\newcommand{\ccase}[1]{%
  \stepcounter{casecounter}%
  \setcounter{subcasecounter}{0}%
  \setcounter{subsubcasecounter}{0}%
  \protected@write \@auxout {}{\string \newlabel {#1}{{\thecasecounter}{\thepage}{\thecasecounter}{#1}{}} }%
  \hypertarget{#1}{\noindent\textbf{Case \thecasecounter.}}
}
\newcommand{\subcase}[1]{%
  \stepcounter{subcasecounter}%
  \setcounter{subsubcasecounter}{0}%
  \protected@write \@auxout {}{\string \newlabel {#1}{{\thecasecounter.\thesubcasecounter}{\thepage}{\thecasecounter.\thesubcasecounter}{#1}{}} }%
  \hypertarget{#1}{\noindent\textbf{Case \thecasecounter.\thesubcasecounter.}}
}
\newcommand{\subsucase}[1]{%
  \stepcounter{subsubcasecounter}%
  \protected@write \@auxout {}{\string \newlabel {#1}{{\thecasecounter.\thesubcasecounter.\thesubsubcasecounter}{\thepage}{\thecasecounter.\thesubcasecounter.\thesubsubcasecounter}{#1}{}} }%
  \hypertarget{#1}{\noindent\textbf{Case \thecasecounter.\thesubcasecounter.\thesubsubcasecounter.}}
}
\makeatother

\title{On Mixed Linear Layouts of\\Series-Parallel Graphs}

\author{Patrizio~Angelini\inst{1}\orcidID{0000-0002-7602-1524} 
\and Michael~A.~Bekos\inst{2,5}\orcidID{0000-0002-3414-7444}
\and Philipp~Kindermann\inst{3,5}\orcidID{0000-0001-5764-7719}
\and Tamara~Mchedlidze\inst{4}\orcidID{0000-0002-1545-5580}}
\institute{%
  John Cabot University, Rome, Italy,
  \email{pangelini@johncabot.edu}
  \and
  Universit\"at T\"ubingen, Germany,
  \email{bekos@informatik.uni-tuebingen.de}
  \and
  Universit\"at W\"urzburg, Germany,
  \email{philipp.kindermann@uni-wuerzburg.de}
  \and
  Karlsruhe Institute of Technology (KIT), Germany,
  \email{mched@iti.uka.de}
  \and
  Universität Passau, Germany}

\authorrunning{P. Angelini et al.}
\titlerunning{On Mixed Linear Layouts of Series-Parallel Graphs} 

\newcommand{\layout}[2]{mixed #1-stack #2-queue layout\xspace}
\newcommand{\layouts}[2]{mixed #1-stack #2-queue layouts\xspace}
\newcommand{\graph}[2]{\ensuremath G(#1,#2)}

\newcommand{\mixed}{mixed layout\xspace}
\newcommand{\mixeds}{mixed layouts\xspace}

\begin{document}
\maketitle

\begin{abstract}
  A mixed $s$-stack $q$-queue layout of a graph consists of a linear order
  of its vertices and of a partition of its edges
  into $s$ stacks and $q$ queues, such that no two edges in the same stack 
  cross and no two edges in the same queue nest.
  In 1992, Heath and Rosenberg conjectured that every planar graph admits a 
  mixed 1-stack 1-queue layout. 
  Recently, Pupyrev disproved this conjectured by demonstrating a planar 
  partial 3-tree that does not admit a 1-stack 1-queue layout.
  In this note, we strengthen Pupyrev's result by showing that the conjecture
  does not hold even for 2-trees, also known as series-parallel graphs.    
  \keywords{mixed linear layouts, queue layouts, book embeddings, series-parallel graphs}
\end{abstract}

\section{Introduction}
\label{sec:introduction}
Over the years, linear layouts of graphs have been a fruitful subject of intense research, which has resulted in several remarkable results both of combinatorial and of algorithmic nature; see, e.g.,~\cite{DBLP:journals/jct/BernhartK79,DBLP:conf/focs/DujmovicJMMUW19,DBLP:conf/focs/Heath84,DBLP:journals/siamcomp/HeathR92,DBLP:journals/combinatorics/Wiechert17,DBLP:journals/jcss/Yannakakis89}.
A linear layout of graph is defined by a total order of its vertex-set and by a partition of its edge-set into a number of subsets, called \df{pages}.
By imposing different constraints on the edges that may reside in the same page, one obtains different types of linear layouts; see~\cite{DBLP:conf/wg/AlamBG0P18,DBLP:journals/ejc/BinucciGHL18,DBLP:journals/siamcomp/HeathR92,DBLP:conf/gd/Pupyrev17,DBLP:journals/jcss/Yannakakis89}. The most notable ones are arguably the stack and the queue layouts (the former are commonly referred to as \df{book embeddings} in the literature), as is evident from the numerous papers that have been published over the years; see~\cite{DBLP:journals/dmtcs/DujmovicW04} for a short introduction. 	

In a \df{stack} (\df{queue}) \df{layout} of a graph, no two indepedent edges of the same page, called \df{stack} (\df{queue}) in this context, are allowed to cross (nest, resp.) with respect to the underlying linear order; see~\cite{DBLP:journals/jct/BernhartK79} and~\cite{DBLP:journals/siamcomp/HeathR92}. In other words, the endpoints of the edges assigned to the same stack follow the last-in-first-out model in the underlying linear order, while the endpoints of the edges assigned to the same queue follow the first-in-first-out model; see~Fig.~\ref{fig:gh}. 
The minimum number of stacks (queues) required by any of the stack (queue) layouts of a graph is commonly referred to as its \df{stack-number} (\df{queue-number}, resp.).  
Accordingly, the stack-number (queue-number) of a class of graphs is the maximum stack-number (queue-number, resp.) over all its members.

\begin{figure}[t]
	\centering
	\subcaptionbox{\label{fig:stack}$2$-stack layout}{\includegraphics[page=1,width=0.32\textwidth]{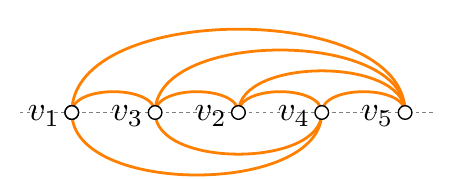}}
	\subcaptionbox{\label{fig:queue}$2$-queue layout}{\includegraphics[page=2,width=0.32\textwidth]{introduction}}
	\subcaptionbox{\label{fig:mixed}mixed $1$-stack $1$-queue}{\includegraphics[page=3,width=0.32\textwidth]{introduction}}
   \caption{%
   Illustration of different linear layouts of the complete graph on five vertices $v_1,\ldots,v_5$ minus the edge $(v_1,v_2)$.}
\label{fig:gh}
\end{figure}

\medskip\noindent\textbf{Known Results.} A large body of the literature is devoted to the study of 
bounds on the stack- and the queue-number of different classes of graphs. 

For stack layouts, the most remarkable result is due to Yannakakis, who back in 1986 showed that every planar graph admits a $4$-stack layout~\cite{DBLP:conf/stoc/Yannakakis86,DBLP:journals/jcss/Yannakakis89}. 
Recently, Bekos et al.~\cite{DBLP:journals/corr/abs-2004-07630} and Yannakakis~\cite{DBLP:journals/corr/abs-2004-01348} independently established that the stack-number of the class of planar graphs is $4$, by demonstrating planar graphs that do not admit $3$-stack layouts. Certain subclasses of planar graphs, however, allow for layouts with fewer than four stacks, e.g., 
$4$-connected planar graphs~\cite{NC08}, 
series-parallel graphs~\cite{DBLP:conf/cocoon/RengarajanM95}, 
planar $3$-trees~\cite{DBLP:conf/focs/Heath84}, 
and others~\cite{DBLP:journals/algorithmica/BekosGR16,DBLP:journals/mp/CornuejolsNP83,Ewald1973,DBLP:journals/dcg/FraysseixMP95,DBLP:journals/dam/GuanY2019,DBLP:conf/esa/0001K19,DBLP:journals/appml/KainenO07}.

For queue layouts, Dujmovi\'c et al.~\cite{DBLP:conf/focs/DujmovicJMMUW19} recently 
showed that every planar graph admits a $49$-queue layout, improving over previously known
logarithmic bounds~\cite{DBLP:journals/corr/BannisterDDEW18,DBLP:journals/siamcomp/BattistaFP13,DBLP:journals/jct/Dujmovic15,DBLP:journals/jgaa/DujmovicF18}. 
However, the exact queue-number of the class of planar graphs is not yet known, 
as the currently best-known lower bound is $4$~\cite{DBLP:conf/gd/AlamBG0P18}. Again, 
several subclasses of planar graphs allow for layouts with significantly fewer 
than $49$ queues, e.g., 
outerplanar  graphs~\cite{DBLP:journals/siamdm/HeathLR92}, 
series-parallel graphs~\cite{DBLP:conf/cocoon/RengarajanM95} and 
planar $3$-trees~\cite{DBLP:conf/gd/AlamBG0P18}.

\medskip\noindent\textbf{Motivation.} Back in 1992, Heath and Rosenberg~\cite{DBLP:journals/siamcomp/HeathR92}
 proposed a natural generalization of stack and queue layouts, called \df{\layout{$s$}{$q$}},
that supports  $s$ stack-pages and $q$ queue-pages.
In their seminal paper~\cite{DBLP:journals/siamcomp/HeathR92}, they
conjectured that every planar graph admits a \layout{$1$}{$1$}. However, 
Pupyrev~\cite{DBLP:conf/gd/Pupyrev17} recently showed that the conjecture does not hold even for partial planar $3$-trees.
This negative result naturally raises the question whether the conjecture 
holds for other subclasses of planar graphs. To this end,
 Pupyrev 
 conjectured that bipartite planar graphs admit \layouts{$1$}{$1$}.

\medskip\noindent\textbf{Our contribution.} We make a step forward 
in understanding which subclasses of planar graphs admit \layouts{$1$}{$1$}
by providing a negative certificate for the class of $2$-trees (also known as maximal series-parallel graphs). 
This improves upon the partial planar $3$-tree negative example by 
Pupyrev~\cite{DBLP:conf/gd/Pupyrev17}. Note that $2$-trees admit both $2$-stack 
layouts and $3$-queue layouts~\cite{DBLP:conf/cocoon/RengarajanM95}.


\medskip\noindent\textbf{Preliminaries.}
A \df{linear order} $\prec$ of a graph $G$ is a total order of its vertices. 
Let  $F=\{(u_i,v_i);\;i=1,\ldots,k\}$ be a set of $k \geq 2$ independent~edges such that $u_i \prec v_i$, for all $1 \leq i \leq k$. If the order is $[u_1, \ldots, u_k, v_k, \ldots, v_1]$, then we say that the edges of $F$ form a \df{$k$-rainbow}, while if the order is $[u_1, \ldots, u_k, v_1, \ldots, v_k]$, then the edges of $F$ form a \df{$k$-twist}. 
Two edges that form a $2$-twist ($2$-rainbow) are referred to as \df{crossing} (\df{nested}, resp.). A \df{stack} (\df{queue}) is a set of pairwise non-crossing (non-nested, resp.) edges. A \df{\layout{$s$}{$q$}} $\cal L$ of $G$ consists of a linear order $\prec$ of $G$ and a partition of the edges of $G$ into $s$ stacks and $q$ queues; for short, we refer to $\cal L$ as \df{\mixed} when $s=q=1$. An edge in a stack (queue) in $\cal L$ is called a \df{stack-edge} (\df{queue-edge},~resp.).

\begin{figure}[t]
	\centering
	\subcaptionbox{\label{fig:rainbow}$3$-rainbow}{\includegraphics[page=1,width=0.32\textwidth]{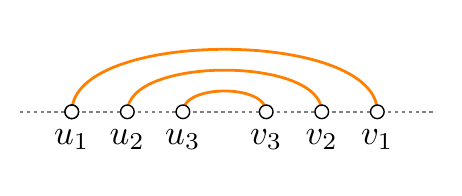}}
  \hfil
	\subcaptionbox{\label{fig:twist}$3$-twist}{\includegraphics[page=2,width=0.32\textwidth]{preliminaries}}
  \hfil
	\subcaptionbox{\label{fig:smiley-face-def}Smiley face}{\includegraphics[page=4,width=0.32\textwidth]{preliminaries}}
   \caption{%
   Illustration of:
   (a)~a $3$-rainbow,
   (b)~a $3$-twist, and
   (c)~a smiley face.}
\label{fig:preliminaries}
\end{figure}

The operation of \df{attaching} a vertex $u$ to an edge $(v,w)$ of a graph $G$
 consists of adding to $G$ vertex $u$ and edges $(u,v)$ and $(u,w)$. Vertex $u$ 
is said to be \df{attached} or being an \df{attachment} of $(v,w)$. 
A \df{$2$-tree} is a graph 
obtained from an edge by repeatedly attaching a vertex to an edge.
Consider a \layout{$s$}{$q$} $\cal L$ of a $2$-tree. We say 
that a vertex $u$ attached to an edge $(v,w)$ is a \df{stack-attachment} 
(\df{queue-attachment}) of $(v,w)$ if both $(u,v)$ and $(u,w)$ are stack-edges 
(queue-edges, resp.) in $\cal L$. Vertex $u$ is a \df{mixed-attachment} of 
$(v,w)$ if one of $(u,v)$ and $(u,w)$ is a queue-edge and the other is a 
stack-edge in $\cal L$.

\section{The Main Result}
\label{sec:series-parallel}

In this section, we define a family $\{G(k,\ell);\; k,\ell\in \mathbb{N}^+\}$ of $2$-trees, and we prove that infinitely many members of it do not admit \mixeds. For $\ell \geq 1$, $G(1,\ell)$ is an edge; for $k>1$, $G(k,\ell)$ is obtained from $G(k-1,\ell)$ by attaching $\ell$ vertices to each edge of it. For convenience, we let $\overline{G}(k,\ell)$ be the graph $G(k,\ell) \setminus G(k-1,\ell)$, that is, the graph induced by the edges that belong to $G(k,\ell)$ but not to $G(k-1,\ell)$.
In the following lemmas, we study properties of a \mixed of graph $\graph{k}{\ell}$.

\begin{lemma}\label{lem:two-stack}
Let $\cal L$ be a \mixed of $\graph{k}{\ell}$ with $k>1,\ell>2$.
Then, every edge of $\graph{k-1}{\ell}$ has at most two stack-attachments in $\cal L$.
\end{lemma}
\begin{proof}
Let $(a,b)$ be an edge of $\graph{k-1}{\ell}$ and assume to the contrary that~there exist three stack-attachments $u$, $v$ and $w$ of $\overline{G}(k,\ell)$ attached to $(a,b)$ in $\cal L$. Neglecting edge $(a,b)$, vertices $a$, $b$, $u$, $v$ and $w$ induce a $K_{2,3}$ in $\graph{k}{\ell}$, whose edges are all stack-edges in $\cal L$. This is a contradiction, since the subgraph induced by the stack-edges of $\graph{k}{\ell}$ must be outerplanar~\cite{DBLP:journals/jct/BernhartK79}, while $K_{2,3}$ is not. 
\end{proof}
%
A \df{smiley face} $\langle a, b, u, v, c, d \rangle$ in a \mixed 
consists of six vertices $a\prec b\prec u\prec v\prec c\prec d$ and four 
edges $(a,b)$, $(c,d)$, $(a,d)$, and $(u,v)$, such that $(a,b)$, $(c,d)$, and 
$(a,d)$ are queue-edges, and thus $(u,v)$ is a stack-edge; 
see Fig.~\ref{fig:smiley-face-def}.

%

\begin{lemma}\label{lem:smiley}
Let $\cal L$ be a \mixed of $\graph{k}{\ell}$ with $k>1,\ell>2$.
Then, a smiley face cannot be formed by the vertices of $\graph{k-1}{\ell}$ in $\cal L$.
\end{lemma}
\begin{proof}
Assume to the contrary that a smiley face $\langle a,b,u,v,c,d\rangle$ is formed in~$\cal L$ by vertices of $\graph{k-1}{\ell}$. 
Consider any vertex $x$ of $\overline{G}(k,\ell)$ attached to the stack-edge $(u,v)$. 
If~$a\prec x\prec d$, then the queue-edge $(a,d)$ forms a 2-rainbow both with $(u,x)$ and with $(v,x)$; see Fig.~\ref{fig:smiley-face-inside}. 
If $x \prec a$, then the queue-edge $(a,b)$ forms a 2-rainbow both with $(u,x)$ and with $(v,x)$; see Fig.~\ref{fig:smiley-face-outside}. 
If $d \prec x$, then the queue-edge $(c,d)$ forms a 2-rainbow both with $(u,x)$ and with $(v,x)$.
Hence, neither $(u,x)$ nor $(v,x)$ is a queue-edge, 
so $x$ is a stack-attachment.
Since $\ell>2$, $(u,v)$  has more than two stack-attachments in~$\cal L$, contradicting Lemma~\ref{lem:two-stack}.
\end{proof}

\begin{figure}[t]    
  \subcaptionbox{\label{fig:smiley-face-inside}}{\includegraphics[page=5,width=0.28\textwidth]{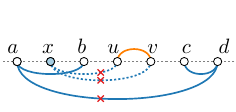}}
  \hfill
  \subcaptionbox{\label{fig:smiley-face-outside}}{\includegraphics[page=6,width=0.28\textwidth]{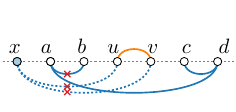}}
  \hfill
  \subcaptionbox{\label{fig:queue-in-the-middle-left}}{\includegraphics[page=5,width=.16\textwidth]{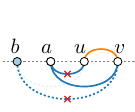}}
  \hfill
  \subcaptionbox{\label{fig:queue-in-the-middle-right}}{\includegraphics[page=6,width=.2\textwidth]{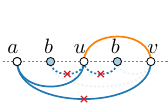}}
  \caption{Illustrations for the proofs 
  (a--b)~of Lemma~\ref{lem:smiley}, and
  (c--d)~of Lemma~\ref{lem:queue-in-the-middle}.}%
  \label{fig:smiley-queue-in-the-middle}%
\end{figure}

\begin{lemma}\label{lem:queue-in-the-middle}
Let $\cal L$ be a \mixed of $\graph{k}{\ell}$ with $k>1,\ell>2$.
Let $a,b,c$~be queue-attachments of an edge $(u,v)$ of $\graph{k-1}{\ell}$ with $u\prec v$. 
Then \mbox{$u \prec a,b,c\prec v$.}
\end{lemma}
\begin{proof}\
Assume to the contrary that $a \prec u$ (the case $v \prec a$ is symmetric).
We first prove that $a\prec u$ implies $v\prec b,c$.
Indeed, if $b \prec a$, then the queue-edges $(b,v)$ and $(a,u)$ form a 2-rainbow; see Fig.~\ref{fig:queue-in-the-middle-left}.
If $a\prec b\prec v$, then the queue-edges $(a,v)$ and $(b,u)$ form a 2-rainbow; see Fig.~\ref{fig:queue-in-the-middle-right}.
Thus, $v\prec b$ and analogously $v\prec c$.
Symmetrically, $v \prec c$ implies $b\prec u$. 
Hence, $b\prec u\prec v\prec b$; a contradiction.
\end{proof}

\begin{lemma}\label{lem:six-queue}
  Let $\cal L$ be a \mixed of $\graph{k}{\ell}$ with $k>4,\ell>6$.
  Then, every queue-edge of $\graph{k-3}{\ell}$ has at most six queue-attachments in $\cal L$.
\end{lemma}
\begin{proof}
  Assume for a contradiction that there is a queue-edge $(u,v)$ in $\graph{k-3}{\ell}$
  with seven queue-attachments $x_1,\ldots,x_7$ in $\overline{G}(k-2,\ell)$.
  By Lemma~\ref{lem:queue-in-the-middle},
  all seven vertices have to lie between~$u$ and~$v$; w.l.o.g. assume that 
  $u\prec x_1\prec\ldots\prec x_7\prec v$.

  \begin{figure}[t]
      \subcaptionbox{\label{fig:six-queue-queue}$(w,x_i)$ queue-edge}{\includegraphics[page=1,width=.32\textwidth]{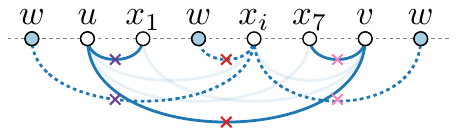}}
    \hfil
    \subcaptionbox{\label{fig:six-queue-before-u}Claim~\ref{claim:before-u}: $w\prec u$}{\includegraphics[page=2,width=0.32\textwidth]{queue-at-most-six-times}}
    \hfil
    \subcaptionbox{\label{fig:six-queue-between-x1-x7}Claim~\ref{claim:between-x1-x7}: $x_1\prec w\prec x_7$}{\includegraphics[page=3,width=.32\textwidth]{queue-at-most-six-times}}

   \subcaptionbox{\label{fig:six-queue-between-x7-v}Claim~\ref{claim:between-x7-v}: $x_7\prec w\prec v$}{\includegraphics[page=4,width=.32\textwidth]{queue-at-most-six-times}}
    \hfil
    \subcaptionbox{\label{fig:six-queue-between-u-x1}Claim~\ref{claim:between-u-x1}: $u\prec w\prec x_1$}{\includegraphics[page=5,width=0.32\textwidth]{queue-at-most-six-times}}
    \hfil
    \subcaptionbox{\label{fig:six-queue-after-v}$v\prec w$}{\includegraphics[page=6,width=.32\textwidth]{queue-at-most-six-times}}
    \caption{Illustrations for the proof of Lemma~\ref{lem:six-queue}.}
    \label{fig:six-queue}
  \end{figure}

  
  For any edge $(u,x_i)$ or $(v,x_i)$ with $2\le i\le 6$ belonging to $\overline{G}(k-1,\ell)$, 
  consider an attachment~$w$ of this edge.
  By Lemma~\ref{lem:two-stack}, we can assume that~$w$ is not a stack attachment.
  Further, if $(w,x_i)$ is a queue-edge, then it forms a 2-rainbow
  with either $(u,v)$
  , $(u,x_1)$
  , or $(v,x_7)$
  ; see Fig.~\ref{fig:six-queue-queue}. Hence, we assume that every selected attachment $w$ 
  of $(u,x_i)$ or $(v,x_i)$ with $2\le i\le 6$ in $\overline{G}(k-1,\ell)$
  is a mixed-attachment with stack-edge $(w,x_i)$.
  We prove Claims \ref{claim:before-u}--\ref{claim:between-u-x1} for edges $(v,x_i)$; for
  $(u,x_i)$ symmetric arguments work; see Fig.~\ref{fig:six-queue}.
  
  \begin{myclaim}\label{claim:before-u}
    There is no mixed-attachment $w$ of $(v,x_i)$ with $2\le i\le 6$ and $w\prec u$
    and there is no mixed-attachment $w$ of $(u,x_i)$ with $2\le i\le 6$ and $v\prec w$.
  \end{myclaim}
  \begin{proof}
    Otherwise, the queue-edges $(v,w)$ and $(u,x_1)$ form a 2-rainbow.
  \end{proof}


  \begin{myclaim}\label{claim:between-x1-x7}
    There is no mixed-attachment $w$ of $(v,x_i)$ or $(u,x_i)$ with $2\le i\le 6$ and $x_1\prec w\prec x_7$.
  \end{myclaim}
  \begin{proof}
    Otherwise, there is a smiley face $\langle u,x_1,x_i,w,x_7,v\rangle$
    or $\langle u,x_1,w,x_i,x_7,v\rangle$ in $\graph{k-1}{\ell}$, based on whether $x_i\prec w$ 
    or $w\prec x_i$, contradicting Lemma~\ref{lem:smiley}.
  \end{proof}
  

  \begin{myclaim}\label{claim:between-x7-v}
    There is no mixed-attachment $w$ of $(v,x_i)$ with $2\le i\le 6$ and $x_7\prec w\prec v$
    and no mixed-attachment $w$ of $(u,x_i)$ with $2\le i\le 6$ and $u\prec w\prec x_1$.
  \end{myclaim}
  \begin{proof}
    Let to the contrary~$w'$ be a mixed-attachment of $(v,x_{i+1})$.
    We have $x_i\prec w'\prec w$,
    as otherwise the stack-edges $(w',x_{i+1})$ and $(x_i,w)$ would cross.
    Then~a smiley face $\langle u,x_1,x_{i+1},w',w,v\rangle$ 
    exists in $\graph{k-1}{\ell}$, contradicting~Lemma~\ref{lem:smiley}.
    %
  \end{proof}
  
  %
  
  \begin{myclaim}\label{claim:between-u-x1}
    There is no mixed-attachment $w$ of $(v,x_i)$ with $3\le i\le 5$ and $u\prec w\prec x_1$
    and no mixed-attachment $w$ of $(u,x_i)$ with $3\le i\le 5$ and $x_7\prec w\prec v$.
  \end{myclaim}
  \begin{proof}
    Let to the contrary~$w'$ be a mixed-attachment of $(u,x_{i-1})$.
    We have $u\prec w\prec w'\prec x_i$,
    as otherwise the stack-edges $(w',x_{i-1})$ and $(x_i,w)$ would cross.
    However, by Claims~\ref{claim:between-x1-x7} and~\ref{claim:between-x7-v},
    this leads to a contradiction.
    %
  \end{proof}
  
  
  
\noindent  Now consider a mixed-attachment $w$ of $(v,x_4)$ and a mixed-attachment~$w'$ of $(u,x_5)$.
  By Claims~\ref{claim:before-u}--\ref{claim:between-u-x1}, we must have $v\prec w$
  and $w'\prec u$; see Fig.~\ref{fig:six-queue-after-v}. 
  However, then the stack-edges $(x_4,w)$ and $(x_5,w')$ cross. 
  This concludes the proof.
\end{proof}

\noindent Lemmas~\ref{lem:two-stack} and~\ref{lem:six-queue} imply the following 

\begin{corollary}\label{cor:mixed-attachments}
Let $\cal L$ be a \mixed of $\graph{k}{\ell}$ with $k>4,\ell>8$.
  Then, every queue-edge of $\graph{k-4}{\ell}$ has at least $\ell-8$ mixed-attachments in~$\cal L$.  
\end{corollary}

\noindent Next we define three patterns \ref{pattern-1}--\ref{pattern-2} and prove that they are forbidden in a \mixed. Each pattern is denoted by $\langle p_1,\ldots,p_7\rangle$, as it is defined on a set of seven vertices for which either $p_1\prec\ldots\prec p_7$ or $p_7\prec\ldots \prec p_1$ holds in $\cal L$; see Fig.~\ref{fig:patterns}. The involved edges in each pattern and their types are as follows.

\begin{figure}[t]
    \subcaptionbox{\label{fig:forbidden-pattern-i-def}Pattern~\ref{pattern-1}}{\includegraphics[page=5,width=.32\textwidth]{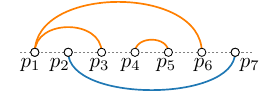}}
  \hfill
  \subcaptionbox{\label{fig:forbidden-pattern-ia-def}Pattern~\ref{pattern-1a}}{\includegraphics[page=9,width=0.32\textwidth]{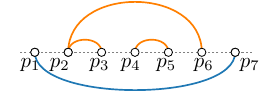}}
  \hfill
  \subcaptionbox{\label{fig:forbidden-pattern-ii-def}Pattern~\ref{pattern-2}}{\includegraphics[page=10,width=.32\textwidth]{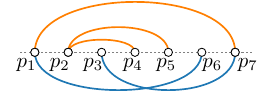}}
  \hfill
  \caption{Illustration of different patterns.}
  \label{fig:patterns}
\end{figure}

\medskip
\begin{enumerate}[label={\bf P.\arabic*}]
\item \label{pattern-1}
Stack-edges $(p_1,p_3)$, $(p_1,p_6)$ and $(p_4,p_5)$, and a queue-edge $(p_2,p_7)$.
\item \label{pattern-1a}
Stack-edges $(p_2,p_3)$, $(p_2,p_6)$ and $(p_4,p_5)$, and a queue-edge $(p_1,p_7)$.
\item \label{pattern-2}
Stack-edges $(p_1,p_7)$, $(p_2,p_4)$ and $(p_2,p_5)$, and queue-edges $(p_1,p_6)$ and $(p_3,p_7)$.
\end{enumerate}

\newcommand{\forbpat}{Let $\cal L$ be a \mixed of $\graph{k}{\ell}$ with $k>1,\ell>4$.  
Then, $\graph{k-1}{\ell}$ does not contain Patterns~\ref{pattern-1}--\ref{pattern-2} in $\cal L$.}
\wormhole{forbpat}

\begin{lemma}\label{lem:forbidden-pattern-i}
  \forbpat
\end{lemma}
\begin{sketch}
For a contradiction, let $\langle p_1,\ldots,p_7\rangle$ be Pattern~\ref{pattern-1} contained in $\graph{k-1}{\ell}$; see Fig.~\ref{fig:forbidden-pattern-i-contr}.
  We first argue that at least one of the $\ell>4$ vertices attached to $(p_4,p_5)$ in $\overline{G}(k,\ell)$ has to be a mixed-attachment.
  By Lemma~\ref{lem:two-stack}, at most two of them can be stack-attachments.
  If more than two of these vertices~are queue-attachments, then by Lemma~\ref{lem:queue-in-the-middle}, they all appear between $p_4$ and~$p_5$ in~$\cal L$,
  and thus any queue-edge incident to them creates a 2-rainbow with the queue-edge $(p_2,p_7)$.
  Hence, there is at least~one mixed-attachment~$x$ of $(p_4,p_5)$.
  Let $e$ and $e'$ be the stack- and queue-edge incident to $x$, respectively.
  Then, $p_3\prec x\prec p_6$, as
  otherwise~$e$ would cross one of the stack-edges $(p_1,p_3)$ and $(p_1,p_6)$.
  However, then~$e'$ forms a 2-rainbow with the queue-edge $(p_2,p_7)$;
  a contradiction. Similarly we argue for Pattern~\ref{pattern-1a}.
  For Pattern~\ref{pattern-2} see the appendix.
\end{sketch}

\begin{figure}[b]
  \centering
  \includegraphics[page=6,width=.32\textwidth]{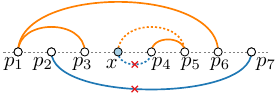}  
  \caption{Illustration for the proof
  of Pattern~\ref{pattern-1} in Lemma~\ref{lem:forbidden-pattern-i}.}
  \label{fig:forbidden-pattern-i-contr}
\end{figure}

\noindent We are now ready to prove the main result of this paper.

\newcommand{\dep}{33\xspace}
\newcommand{\main}{$\graph{k}{\ell}$ does not admit a \mixed if $k\ge 5,\ell\ge \dep$.}
\wormhole{main}
\begin{theorem}\label{thm:main}
\main  
\end{theorem}
\begin{sketch}
Assume to the contrary that $\graph{5}{\dep}$ admits a \mixed $\cal L$.
By Lemma~\ref{lem:two-stack}, there is at least one queue-edge $(u,v)$ in $\graph{2}{\dep}$. W.l.o.g., let $u \prec v$ in $\cal L$. 
By Corollary~\ref{cor:mixed-attachments}, $\graph{3}{\dep}$ contains at least 25 mixed-attachments, say $x_1,\ldots,x_{25}$,~of~$(u,v)$.
For every $i=1,\ldots,25$, one of the following applies: $x_i \prec u$, or $u \prec x_i \prec v$, or $v \prec x_i$. For each of the cases, we further distinguish whether the edge $(u,x_i)$ is a stack-edge or a queue-edge. This defines six configurations for $x_i$. 
Thus, at least five vertices, say w.l.o.g., $x_1,\ldots,x_5$, are attached with 
the same configuration to $(u,v)$; we assume w.l.o.g. that $x_1 \prec \ldots \prec x_5$. 
We show a contradiction in the case when $v \prec x_i$  
and $(u,x_i)$ is a queue-edge for all $i=1,\ldots,5$; the remaining cases are in the appendix.

\begin{figure}[t]
  \subcaptionbox{\label{fig:configuration-1-1}}{\includegraphics[page=7,width=.32\textwidth]{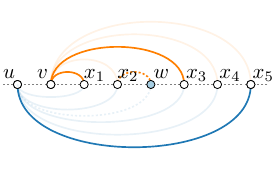}}
  \hfill
  \subcaptionbox{\label{fig:configuration-1-2}}{\includegraphics[page=9,width=.32\textwidth]{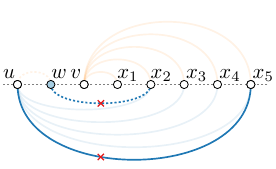}}
  \hfill
  \subcaptionbox{\label{fig:configuration-1-3}}{\includegraphics[page=10,width=.32\textwidth]{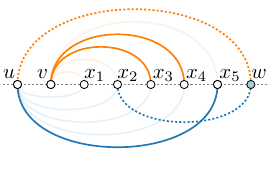}}
  \caption{Illustration for the first case of Theorem~\ref{thm:main}.}
  \label{fig:configuration-1}
\end{figure}

By Corollary~\ref{cor:mixed-attachments}, $\graph{4}{\dep}$ contains at least one mixed-attachment~$w$ of $(u,x_2)$. Thus, either $(x_2,w)$ or $(u,w)$ is a stack-edge. 
In the former case, the stack-edges $(v,x_1)$ and $(v,x_3)$ enforce $x_1\prec w\prec x_3$; 
see Fig.~\ref{fig:configuration-1-1}. 
Hence, $\langle u,v,x_1,x_2,w,x_3,x_5\rangle$ or $\langle u,v,x_1,w,x_2,x_3,x_5\rangle$ 
of $\graph{4}{\dep}$ form  Pattern~\ref{pattern-1a} in $\cal L$.
This contradicts Lemma~\ref{lem:forbidden-pattern-i}. 
In the latter case, the stack-edge $(v,x_5)$ enforces either $w\prec v$ or $x_5\prec w$. 
We consider three subcases. If $w\prec u$, then the queue-edges $(w,x_2)$ and
 $(u,x_1)$ form a 2-rainbow. If $u\prec w\prec v$, then the queue-edges
 $(w,x_2)$ and $(u,x_5)$ form a 2-rainbow; see Fig.~\ref{fig:configuration-1-2}.
Otherwise, $x_5\prec w$ holds. It follows that $\langle u,v,x_2,x_3,x_4,x_5,w\rangle$  
of $\graph{4}{\dep}$ form Pattern~\ref{pattern-2} in $\cal L$; see Fig.~\ref{fig:configuration-1-3}.
\end{sketch}

\section{Open Problems}
\label{sec:conclusions}

In this paper, we proved that $2$-trees do not admit \layouts{1}{1}. Since $2$-trees admit $2$-stack layouts and $3$-queue layouts~\cite{DBLP:conf/cocoon/RengarajanM95}, it is natural to ask whether they admit \layouts{1}{2}. We conclude with an algorithmic question, namely, what is the complexity of recognizing graphs that admit \layouts{1}{1}, even for $2$-trees?
Note that recently de Col et al.~\cite{DBLP:conf/gd/ColKN19} showed that 
 testing whether a (not necessarily planar) 
graph admits a \layout{$2$}{$1$} is NP-complete.

\bibliographystyle{splncs03}
\bibliography{abbrv,stacks,queues}

\clearpage
\appendix
\section*{\LARGE Appendix}

\input{appendix.tex}

\end{document}

%% file: appendix.tex
In this appendix, we give proofs that were omitted in the main part due to space constraints.

\begin{backInTime}{forbpat}
\begin{lemma}
\forbpat
\end{lemma}
\begin{proof}
We proved in the main part that $\graph{k-1}{\ell}$ does not contain Pattern~\ref{pattern-1}.
We complete the proof of this lemma by showing that $\graph{k-1}{\ell}$ contains neither Pattern~\ref{pattern-1a} nor Pattern~\ref{pattern-2}.

As already mentioned, the proof that $\graph{k-1}{\ell}$ does not contain Pattern~\ref{pattern-1a} is similar to the corresponding one for Pattern~\ref{pattern-1}. Here, we give the proof only for the sake of completeness. For a contradiction, let $\langle p_1,\ldots,p_7\rangle$ be Pattern~\ref{pattern-1a} contained in $\graph{k-1}{\ell}$; see Fig.~\ref{fig:forbidden-pattern-ia-contr}.
%
Consider a mixed-attachment~$x$ of edge $(p_4,p_5)$ in $\overline{G}(k,\ell)$, whose existence is proven based on Lemmas~\ref{lem:two-stack} and~\ref{lem:queue-in-the-middle} as in Pattern~\ref{pattern-1}.
Vertex $x$ has to lie between $p_3$ and $p_6$ in $\cal L$,
  as otherwise the stack-edge incident to $x$ would cross either the stack-edge $(p_2,p_6)$ or the stack-edge $(p_2,p_3)$.
In this case, however, 
  the queue-edge incident to $x$ forms a 2-rainbow with the queue-edge $(p_1,p_7)$; a contradiction.

\begin{figure}[h]
  \centering
  \subcaptionbox{\label{fig:forbidden-pattern-ia-contr}}{\includegraphics[page=12,width=.32\textwidth]{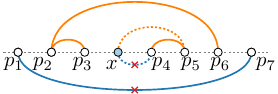}}
  \hfil
  \subcaptionbox{\label{fig:forbidden-pattern-ii-contr}}{\includegraphics[page=11,width=.32\textwidth]{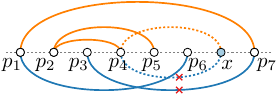}}  
  \caption{Illustrations for the proofs 
  of Patterns~\ref{pattern-1a} and~\ref{pattern-2}.}
  \label{fig:forbidden-patterns-app}
\end{figure}

For a contradiction, let now $\langle p_1,\ldots,p_7\rangle$ be Pattern~\ref{pattern-2} contained in $\graph{k-1}{\ell}$; refer to Fig.~\ref{fig:forbidden-pattern-ii-contr}.
Similar to the proof of Pattern~\ref{pattern-1}, we first argue that at least one of the $\ell>4$ vertices attached to the edge $(p_2,p_4)$ in $\overline{G}(k,\ell)$ has to be a mixed-attachment.
Indeed, by Lemma~\ref{lem:two-stack}, at most two of these vertices can be stack-attachments.
  If more than two of these vertices are queue-attachments, then by Lemma~\ref{lem:queue-in-the-middle} they all appear between $p_2$ and~$p_4$ in $\cal L$,
  which is not possible as any queue-edge incident to them would create a 2-rainbow with the queue-edge $(p_1,p_6)$.
  Hence, at least one vertex~$x$ attached to $(p_2,p_4)$ is a mixed-attachment. Let $e$ and $e'$ be the stack- and queue-edge incident to $x$, respectively.
    Then, $x$ has to lie between $p_1$ and $p_7$ in $\cal L$,
  as otherwise~$e$ would cross the stack-edge $(p_1,p_7)$.
  Also, $x$ cannot lie between $p_1$ and $p_6$, as otherwise
  $e'$ would form a 2-rainbow with the queue-edge $(p_1,p_6)$.
  Hence, $x$ has to lie between $p_6$ and $p_7$ in $\cal L$.
  If the edge $(p_4,x)$ is a queue-edge, i.e., $e'=(p_4,x)$, then it forms a 2-rainbow with the queue-edge $(p_3,p_7)$.
  Otherwise, the edge $(p_4,x)$ is a stack-edge, i.e., $e=(p_4,x)$, which implies that it crosses the stack-edge $(p_2,p_5)$.
  In both cases, we have a contradiction.
\end{proof}
\end{backInTime}

\begin{backInTime}{main}
\begin{theorem}
\main
\end{theorem}
\begin{proof}
Assume to the contrary that $\graph{5}{\dep}$ admits a \mixed $\cal L$.
Consider the subgraph $\graph{1}{\dep}$ of $\graph{5}{\dep}$. By definition, this subgraph is a single edge $(a,b)$. By Lemma~\ref{lem:two-stack}, in the subgraph $\graph{2}{\dep}$ of $\graph{5}{\dep}$, which is obtained by attaching \dep vertices to edge $(a,b)$, there is at least one queue-edge $(u,v)$. W.l.o.g., we assume that $u \prec v$ in $\cal L$. 
Consider now the subgraph $\graph{3}{\dep}$ of $\graph{5}{\dep}$. This subgraph contains \dep attachments of edge $(u,v)$. By Corollary~\ref{cor:mixed-attachments}, at least 25 of them are mixed-attachments. Denote them by $x_1,\ldots,x_{25}$. For each vertex $x_i$ with $i=1,\ldots,25$, one of the following applies: $x_i \prec u$, or $u \prec x_i \prec v$, or $v \prec x_i$. For each of them, we further distinguish whether the edge $(u,x_i)$ is a stack or a queue-edge. This defines six possible configurations for vertex $x_i$. By the pigeonhole principle, there exist at least five vertices, say w.l.o.g., $x_1,\ldots,x_5$, that are attached with the same configuration to edge $(u,v)$. In the following, we find a contradiction in each of these configurations, assuming w.l.o.g.\ $x_1\prec x_2\prec x_3\prec x_4\prec x_5$ in $L$.

\ccase{conf:1} \textit{For $i=1,\ldots,5$, $v \prec x_i$ and edge $(u,x_i)$ is a queue-edge}: The subgraph $\graph{4}{\dep}$ of $\graph{5}{\dep}$ contains \dep attachments to the queue-edge $(u,x_2)$. By Corollary~\ref{cor:mixed-attachments}, at least 25 of them are mixed-attachments. Let $w$ be such an attachment. It follows that either $(x_2,w)$ or $(u,w)$ is a stack-edge.

In the former case, the stack-edges $(v,x_1)$ and $(v,x_3)$ enforce $x_1\prec w\prec x_3$; see Fig.~\ref{fig:configuration-1-1}. It follows that $\langle u,v,x_1,w,x_2,x_4,x_5\rangle$ or $\langle u,v,x_1,w,x_2,x_4,x_5\rangle$ of $\graph{4}{\dep}$ form  Pattern~\ref{pattern-1a} in $\cal L$, depending on whether $x_1 \prec w \prec x_2$ or $x_2 \prec w \prec x_3$, respectively. This contradicts Lemma~\ref{lem:forbidden-pattern-i}. 

In the latter case, the stack-edge $(v,x_5)$ enforces that either $w\prec v$ or $x_5\prec w$. We consider three subcases. If $w\prec u$, then a 2-rainbow is formed by the queue-edges $(w,x_2)$ and $(u,x_1)$. If $u\prec w\prec v$, then a 2-rainbow is formed by the queue-edges $(w,x_2)$ and $(u,x_5)$; see Fig.~\ref{fig:configuration-1-2}.
Otherwise, $x_5\prec w$ holds. It follows that $\langle u,v,x_2,x_3,x_4,x_5,w\rangle$  of $\graph{4}{\dep}$ form Pattern~\ref{pattern-2} in $\cal L$; see Fig.~\ref{fig:configuration-1-3}.
All three cases lead to a contradiction.

\ccase{conf:2} \textit{For $i=1,\ldots,5$, $v \prec x_i$ and edge $(u,x_i)$ is a stack-edge}: The subgraph $\graph{4}{\dep}$ of $\graph{5}{\dep}$ contains \dep attachments to the queue-edge $(v,x_3)$. By Corollary~\ref{cor:mixed-attachments}, at least 25 of them are mixed-attachments. Let $w$ be such an attachment. It follows that either $(x_3,w)$ or $(v,w)$ is a stack-edge. 

\begin{figure}[t]
  \subcaptionbox{\label{fig:configuration-2-1}}{\includegraphics[page=6,width=.32\textwidth]{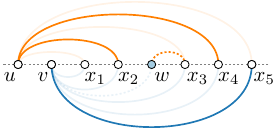}}
  \hfill
  \subcaptionbox{\label{fig:configuration-2-2}}{\includegraphics[page=7,width=.32\textwidth]{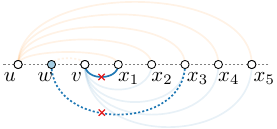}}
  \hfill
  \subcaptionbox{\label{fig:configuration-2-3}}{\includegraphics[page=8,width=.32\textwidth]{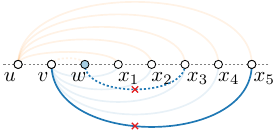}}
  \caption{Illustration for Case~\ref{conf:2} of Theorem~\ref{thm:main}.}
  \label{fig:configuration-2}
\end{figure}

In the former case, the stack-edges $(u,x_2)$ and $(u,x_4)$ enforce $x_2\prec w\prec x_4$; see Fig.~\ref{fig:configuration-2-1}. It follows that $\langle u,v,x_2,w,x_3,x_4,x_5\rangle$ or $\langle u,v,x_2,x_3,w,x_4,x_5\rangle$ of $\graph{4}{\dep}$ form Pattern~\ref{pattern-1} in $\cal L$, depending on whether $x_2 \prec w \prec x_3$ or $x_3 \prec w \prec x_4$, respectively. This contradicts Lemma~\ref{lem:forbidden-pattern-i}. 

In the latter case, the stack-edge $(u,x_1)$ enforces $u\prec w\prec x_1$. We consider two subcases. If $u\prec w\prec v$, then a 2-rainbow is formed by the queue-edges $(w,x_3)$ and $(v,x_1)$; see Fig.~\ref{fig:configuration-2-2}. Otherwise, $v\prec w\prec x_1$ holds, in which case a 2-rainbow is formed by the queue-edges $(v,x_5)$ and $(w,x_3)$; see Fig.~\ref{fig:configuration-2-3}. Both cases lead to a contradiction.

\ccase{conf:3} \textit{For $i=1,\ldots,5$, $u \prec x_i\prec v$ and edge $(u,x_i)$ is a stack-edge}: As in the previous cases, we first observe that the subgraph $\graph{4}{\dep}$ of $\graph{5}{\dep}$ contains \dep attachments to the queue-edge $(v,x_4)$. By Corollary~\ref{cor:mixed-attachments}, at least 25 of them are mixed-attachments. Let $w$ be such an attachment. It follows that either $(x_4,w)$ or $(v,w)$ is a stack-edge. 

In the former case, the stack-edges $(u,x_3)$ and $(u,x_5)$ enforce $x_3\prec w\prec x_5$; see Fig.~\ref{fig:configuration-3-1}. It follows that $\langle u,x_1,x_2,w,x_4,x_5,v\rangle$ or $\langle u,x_1,x_2,x_4,w,x_5,v\rangle$ of $\graph{4}{\dep}$ form Pattern~\ref{pattern-1} in $\cal L$, depending on whether $x_3 \prec w \prec x_4$ or $x_4 \prec w \prec x_5$, respectively. This contradicts Lemma~\ref{lem:forbidden-pattern-i}. 

\begin{figure}[t]
  \subcaptionbox{\label{fig:configuration-3-1}}{\includegraphics[page=7,width=.32\textwidth]{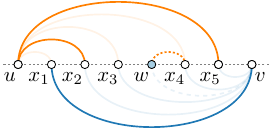}}
  \hfill
  \subcaptionbox{\label{fig:configuration-3-2}}{\includegraphics[page=8,width=.32\textwidth]{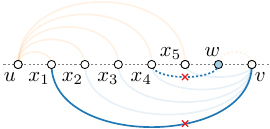}}
  \hfill
  \subcaptionbox{\label{fig:configuration-3-3}}{\includegraphics[page=9,width=.32\textwidth]{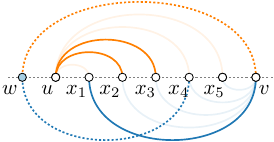}}
  \caption{Illustration for Case~\ref{conf:3} of Theorem~\ref{thm:main}.}
  \label{fig:configuration-3}
\end{figure}

In the latter case, the stack-edge $(u,x_5)$ enforces that either $w\prec u$ or $x_5\prec w$. We consider three subcases. 
If $v\prec w$, then a 2-rainbow is formed by the queue-edges $(w,x_4)$ and $(v,x_5)$. 
If $x_5\prec w\prec v$, then a 2-rainbow is formed by the queue-edges $(w,x_4)$ and $(v,x_1)$; see Fig.~\ref{fig:configuration-3-2}.
Otherwise, $w\prec u$ holds. It follows that $\langle w,u,x_1,x_2,x_3,x_4,v\rangle$  of $\graph{4}{\dep}$ form Pattern~\ref{pattern-2} in $\cal L$; see Fig.~\ref{fig:configuration-3-3}.
All three cases lead to a contradiction.  

\ccase{conf:4} \textit{For $i=1,\ldots,5$, $x_i \prec u$ and edge $(u,x_i)$ is a stack-edge}: This case is symmetric to Case~\ref{conf:1}.
  
\ccase{conf:5} \textit{For $i=1,\ldots,5$, $x_i \prec u$ and edge $(u,x_i)$ is a queue-edge}: This case is symmetric to Case~\ref{conf:2}.

\ccase{conf:6} \textit{For $i=1,\ldots,5$, $u\prec x_i \prec v$ and edge $(u,x_i)$ is a queue-edge}: This case is symmetric to Case~\ref{conf:3}.

\medskip\noindent Since Cases~\ref{conf:1}--\ref{conf:6} have led to a contradiction, $\graph{5}{\dep}$ does not admit any \mixed, as desired.
\end{proof}
\end{backInTime}